\documentclass[12pt,letterpaper]{article}

\usepackage[utf8]{inputenc}

\usepackage{enumerate}
\usepackage{cite}

\usepackage{nameref,hyperref}

\usepackage[right]{lineno}

\usepackage{rotating}
\usepackage{multirow}
\usepackage{booktabs}
\usepackage{float}

\usepackage{microtype}
\DisableLigatures[f]{encoding = *, family = * }

\raggedright
\usepackage{changepage}

\usepackage[aboveskip=1pt,labelfont=bf,labelsep=period,singlelinecheck=off]{caption}

\makeatletter
\renewcommand{\@biblabel}[1]{\quad#1.}
\makeatother

\usepackage{amssymb}
\usepackage{amsthm}
\usepackage{newtxmath}

\newcommand{\bs}[1]{\ensuremath{\boldsymbol{#1}}}

\newcommand{\tp}{{\!\scriptscriptstyle \top}}
\usepackage{pifont}

\newtheorem{prop}{Proposition}
 \usepackage[round]{natbib}
  
\usepackage{color}

\definecolor{Gray}{gray}{.25}

\usepackage{graphicx}

\usepackage{sidecap}

\usepackage{wrapfig}
\usepackage[pscoord]{eso-pic}
\usepackage[fulladjust]{marginnote}
\reversemarginpar

\begin{document}
\vspace*{0.35in}

\begin{flushleft}
{\Large
\textbf\newline{Simple robust genomic prediction and outlier detection for a multi-environmental field trial.}
}
\newline
\\
Emi Tanaka\textsuperscript{1,*}
\\
\bigskip
\bf{1} School of Mathematics and Statistics, The University of Sydney, NSW, Australia, 2006
\\
\bigskip
* emi.tanaka@sydney.edu.au

\end{flushleft}

\section*{Abstract}
The aim of plant breeding trials is often to identify germplasms that are well adapt to target environments. These germplasms are identified through genomic prediction from the analysis of multi-environmental field trial (MET) using linear mixed models. The occurrence of outliers in MET are common and known to adversely impact accuracy of genomic prediction yet the detection of outliers, and subsequently its treatment, are often neglected. A number of reasons stand for this - complex data such as MET give rise to distinct levels of residuals and thus offers additional challenges of an outlier detection method and  many linear mixed model software are ill-equipped for robust prediction. We present outlier detection methods using a holistic approach that borrows the strength across trials. We furthermore evaluate a simple robust genomic prediction that is applicable to any linear mixed model software. These are demonstrated using simulation based on two real bread wheat yield METs with a partially replicated design and an alpha lattice design. 


\section{Introduction \label{intro}}
Multi-environmental trials (METs) are routinely analysed for the evaluation and selection of the best genotypes. These MET data are commonly analysed by linear mixed models with a particular interest in accurate prediction of the main genotype effects or genotype $\times$ environment (G$\times$E) interaction effects. It is widely accepted to use empirical best linear unbiased predictions (E-BLUPs) of G$\times$E effects for the aim of  selection \citep{Robinson1991}. The E-BLUP, however, are sensitive to the presence of outliers resulting in lower accuracy of the prediction of G$\times$E effects \citep{Estaghvirou2014}. Despite the common occurrence of outlying observations (as shown in Figure~\ref{fig:boxplot}), many linear mixed model software are not equipped with robust parameter estimation as well as robust prediction and often rely on the input of a well behaved data. The user may identify potential outliers in a pre-processing step and exclude these from the analysis, however, the exclusion of potential ``real'' observations to fit the model is discouraged. 

\begin{figure}[ht!]
	\centering\includegraphics[width=\textwidth]{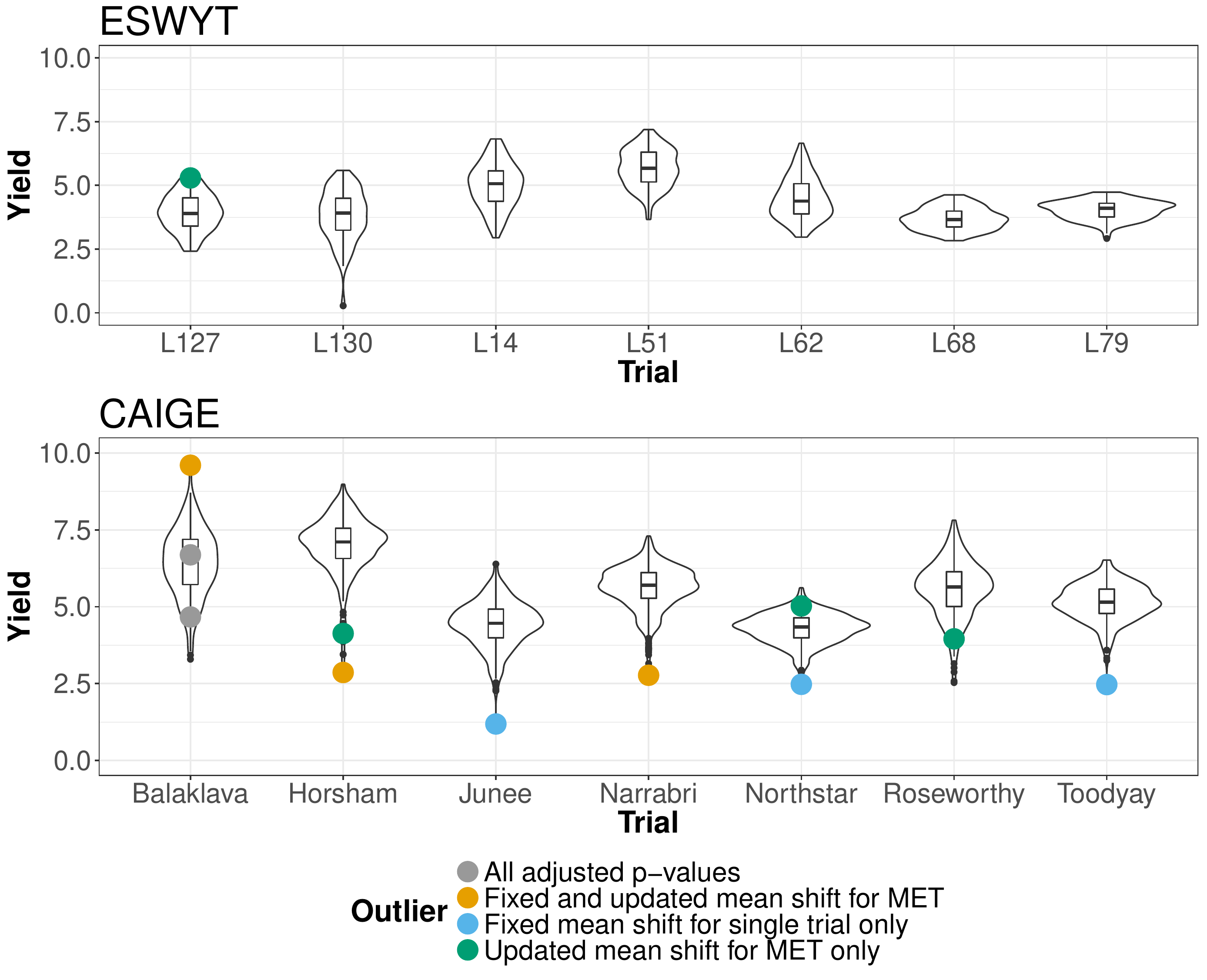}
	\caption{Boxplot of the yield embedded within the violin plot by trial for ESWYT and CAIGE show potential outliers. Further scrutiny of data identify outliers using adjusted $p$-values (Section~\ref{sec:qval}) of mean shift estimated via three MSOM: 1) $t_i$ with baseline model based on single trial analysis; 2) $t_i$ from baseline model as the MET model; and 3) $\hat{\phi}_i$ from MSOM with baseline model as the MET model. The coloured points identify observations flagged as outliers according to the three adjusted $p$-values.}
	\label{fig:boxplot}.
\end{figure}

Outliers give rise to two distinct concepts that are often conflated: outlier detection and  robust estimation. Outlier detection is an important pre-processing step to identify dubious data that may be the result of errors (e.g. transcription error) or other sources that have influenced the observation. Carrying out this step soon after the experiment offers a chance to rectify errors or to identify unexpected factors influencing the observation(s). The information from outlier detection may be carried over for a robust estimation such as by case-deletion or weighted analysis although this is not necessary for some class of robust estimators \citep[e.g. M-estimators, ][]{Huber1981}.

Outlier detection are generally conducted through a form of residual analysis \citep{Schutzenmeister2012} or sensitivity analysis when observations are perturbed or deleted. For single outliers, \cite{Cook1977} proposed the use of a model-based statistics, now widely known as Cook's distance, based on a confidence ellipsoid for the fixed effects. A generalised form of the Cook's distance (GCD) for the linear mixed model was used by \cite{Christensen1992}. Alternative outlier measures were considered by Cook in the linear fixed model, namely, the mean shift outlier model \citep[MSOM,][]{Cook1982a} and the variance shift outlier model \citep[VSOM,][]{Cook1982} where one (or more) observation(s) is considered to arise from a shifted location or inflated variance, respectively, compared to the rest of the observations. For the VSOM, \cite{Cook1982} considered the maximum likelihood estimate of the variance components while the residual maximum likelihood estimate of the variance components \citep{Patterson1971} was considered by \cite{Thompson1985} with further development given in \cite{Gogel1997, Gumedze2008, Gumedze2010} for the linear mixed model.  \cite{Gogel1997, Gumedze2008, Gumedze2010} also considered a more general case of MSOM and VSOM where a (known) group of observations arise from shifted location or inflated variance. The aforementioned methods are described in-depth in Section~\ref{sec:stat}. \cite{Bernal-Vasquez2016} recently considered outlier detection in MET but they did not consider VSOM nor GCD and no explicit connection of studentised conditional residual to MSOM was made.

The main contributions of this article are as follows. We propose outlier detection methods using a holistic approach by borrowing strength across trials and are computationally efficient for practical use. We show empirically that this increases discrimination power and present real examples with potential benefit to our holistic approach. We propose a simple robust genomic prediction for METs and emphasise that this simple robust genomic prediction can be easily incorporated in any linear mixed model software that allows the practitioner to use their preferred software. The effectiveness of the methods are evaluated by simulation from two real bread wheat (\textit{Triticum aestivum}) yield METs that employ $\alpha$-lattice and partial replicated designs. 

\section{Materials and methods}\label{sec:method}
\subsection{Data}
We consider the use of two bread wheat  yield multi-environmental trial data described in Section~\ref{sec:ESWYT} and \ref{sec:CAIGE} which will be referred as ESWYT and CAIGE henceforth. Note that we use environment to be synonymous with a single field trial.

\subsubsection{35th Elite Spring Wheat Yield Trial (ESWYT)} \label{sec:ESWYT}

The 35th Elite Selection Wheat Yield Trial (ESWYT) contains 78 trials testing 50 spring bread wheat germplasm adapted to optimally irrigated, low rainfall areas distributed by CIMMYT to over 20 countries \citep{D35ESWYT}. All trials consist of 100 plots and employ $\alpha$-lattice designs \citep{Patterson1976} with two replication of the 50 germplasms, each replication consisting of 10 blocks of size 5. We select seven trials that correspond to ID 127, 130, 14, 51, 62, 68, and 79 in the ``Occ'' column. These seven were chosen from a pool of trials that have at least a ratio of 1.5 for genotype variance to residual variance based on a single trial analysis.  

\subsubsection{2016 Bread Wheat Yield Trial (CAIGE)}\label{sec:CAIGE}
CIMMYT Australia ICARDA Germplasm Evaluation (CAIGE) project aims to evaluate the performance of international germplasms for local adaptation. The CAIGE bread wheat yield trial in 2016 \citep{CAIGE2016} was conducted at 8 locations (Cadoux, Horsham, Toodyay, Roseworthy, Northstar, Narrabri, Junee and Balaklava) within Australia. The Cadoux site suffered from extensive frost damage and no significant yield information is available. There were 240 germplasms tested across all seven trials with 252-391 plots. Each trial employed a partially replicated ($p$-rep) design \citep{Cullis2006} with two blocks and $p$ ranging from 0.23 to 0.39.

\subsection{Statistical model}
We consider the analysis of a MET data in which $m$ genotypes are grown in $t$ trials (not all genotypes are grown in each trial necessary). Let $\bs{y}_j$ denote a vector of length $n_j$ of phenotypic data for the $j$-th trial for $j=1, \ldots, t$; let $n=\sum_{j=1}^{t} n_j$ be the total number of observations and $\bs{y} = (\bs{y}_1^\tp, \ldots, \bs{y}_t^\tp)^\tp$ be the vector of all phenotypic data. 

In general, we model the MET data $\bs{y}$ as 
\begin{equation}
\bs{y} = \bs{X}_0\bs{\tau}_0 + \bs{Z}_{ge}\bs{u}_{ge} + \bs{Z}_p\bs{u}_p + \bs{e}
\end{equation} 
where $\bs{X}_0$ is the design matrix associated with fixed effects $\bs{\tau}_0 $ that include overall mean and trial effects, $\bs{Z}_{ge}$ is the design matrix associated with random G$\times$E effects $\bs{u}_{ge}$, $\bs{Z}_p$ is the design matrix associated with trial specific random peripheral effects (such as block effects) and $\bs{e}$ is the $n\times 1$ vector of random error. Note that the blocking terms are always included. We assume that $\bs{u}_g$, $\bs{u}_p$ and $\bs{e}$ are independently generated from Gaussian distributions with zero mean and variance matrices $\text{var}(\bs{u}_{ge}) = \bs{G}_{ge}$, $\text{var}(\bs{u}_p) = \bs{G}_p$ and $\text{var}(\bs{e}) = \bs{R}$ (see Section~\ref{sec:varstruc} for specific variance structure used).

\subsubsection{Spatial modelling}

The analysis of a single trial is considered first to identify extraneous variation or spatial trends as described in \citet{Gilmour1997, Stefanova2009} and this is included in either $\bs{\tau}_0$ or $\bs{u}_p$. Spatial modelling was only conducted for CAIGE data as the the spatial coordinates of the plot are not available for ESWYT. The fitted models and its variance estimates for both data are shown in Appendix Table~\ref{tab:var} and \ref{tab:vargxe}.

\subsubsection{Variance structures}\label{sec:varstruc}
For $\bs{G}_{ge}$, a multiplicative structure is assumed $\bs{G}_{ge} = \bs{G}_e \otimes \bs{G}_g$ where $\bs{G}_g$ is a $m\times m$ genotype relationship matrix, assumed in this article to be the identity matrix $\bs{I}_m$, and $\bs{G}_e$ is the $t\times t$ genotype covariance between environment. The structure of $\bs{G}_e$ may be completely general however for large number of environments, a lower order approximation via factor analytic structure \citep{Smith2015a} may only be feasible or more attractive to fit the model. In this article we assume that $\bs{G}_e$ is an unstructured matrix for the MET analysis. 

For this article $\bs{u}_p$ all correspond to trial specific blocking effects such as blocks, replicate, row or column effects. More specifically, suppose that \texttt{Trial}, \texttt{Block}, \texttt{Replicate}, \texttt{Row} and \texttt{Column} are factors that index the trial, block, replicate, row and column, respectively, then following the model syntax of \cite{Wilkinson1973}, the terms in $\bs{u}_p$ include \texttt{Trial|Replicate|Block} and \texttt{Trial|Replicate} for ESWYT and \texttt{Trial|Block}, \texttt{Trial|Row} and \texttt{Trial|Column} for CAIGE. Note this can be completely general to include other non-genetic peripheral effects  as appropriate such as spline effects although the aforementioned effects were deemed sufficient.  The corresponding variance $\bs{G}_p$ is the direct sum of scaled identity matrix $\oplus_{k=1}^q\sigma^2_{pk}\bs{I}_{n_{pk}}$ where $q$ correspond to the number of components in $\bs{u}_p$ and $n_{pk}$ is commensurate with the size of the corresponding random component.  

The structure of $\bs{R}$ is assumed as a direct sum of error variance of each trial, i.e. $\oplus_{j=1}^t \bs{R}_j$. We may further assume $\bs{R}_j$ follows a separable autoregressive process of order one to model local spatial trend as \cite{Smith2001}, however, this is not possible for ESYWT with missing plot coordinates and we found that $\bs{R}_j = \sigma^2_j \bs{I}_{n_j}$ for all seven trials in CAIGE is sufficient.

\subsection{Outlier detection}\label{sec:stat}
Consider a $n \times 1$ vector of response $\bs{y}$ modelled by a general linear mixed model
\begin{equation}
\bs{y} = \bs{X}\bs{\tau} + \bs{Z}\bs{u} + \bs{e} \label{eq:baseline}
\end{equation} 
where $\bs{X}$ is the $n\times p_0$ design matrix of rank $p \leq p_0$ with associated fixed effects $\bs{\tau}$; $\bs{Z}$ is the $n\times q$ design matrix with associated random effects $\bs{u}$ and $\bs{e}$ is the $n\times 1$ vector of random errors. We assume that 
$$
\begin{bmatrix}
\bs{u} \\
\bs{e} 
\end{bmatrix} \sim N\left(\begin{bmatrix}
\bs{0} \\
\bs{0} 
\end{bmatrix}, 
\begin{bmatrix}
\bs{G}(\bs{\kappa}_G) & \bs{0} \\
\bs{0} & \bs{R}(\bs{\kappa}_R)
\end{bmatrix}
\right)
$$
where $\bs{G}$ and $\bs{R}$ are some positive definite matrices that are functions of variance parameters $\bs{\kappa}_G$ and $\bs{\kappa}_R$, respectively.
We refer model \eqref{eq:baseline} as the \emph{baseline model}. For simplicity of later results, we let $\bs{V} = \bs{Z}\bs{G}\bs{Z}^\tp + \bs{R}$; $\bs{P}_X = \bs{V}^{-1} - \bs{V}^{-1} \bs{X}(\bs{X}^\tp\bs{V}^{-1}\bs{X})^-\bs{X}^\tp\bs{V}^{-1}$; $\bs{\kappa} = (\bs{\kappa_G}^\tp, \bs{\kappa_R}^\tp)^\tp$ denote a vector of variance parameters; ${\hat{\bs{\kappa}} }$ is the residual maximum likelihood estimate of $\bs{\kappa}$ from the fit of baseline model \eqref{eq:baseline} and $\hat{\bs{P}}_X = \bs{P}_X(\hat{\bs{\kappa}})$.

The outlier detection methods are described next with the list presented in Table~\ref{tab:method}. We note that the baseline model could be based on either a MET analysis where $\bs{G}_e$ is unstructued or a single trial analysis.

\begin{center}
\begin{table}[ht]
	\centering
	\begin{tabular}{cccc}
		\toprule
		Statistic & Baseline Model & Variance Estimation & Computationally efficient \\ 
		\midrule
		\multirow{2}{*}{$t_i^2$} & Single Trial Analysis & \multirow{2}{*}{Baseline REML} & \multirow{2}{*}{Yes} \\
		& MET Analysis &  &  \\
		\multirow{2}{*}{$s_i^2$} & Single Trial Analysis & \multirow{2}{*}{MSOM REML} & \multirow{2}{*}{No}  \\
		& MET Analysis &   &  \\
		\multirow{2}{*}{$a\text{LRT}_i$} & Single Trial Analysis & \multirow{2}{*}{Baseline REML} &  \multirow{2}{*}{Yes} \\		
		& MET Analysis & &  \\				
		\multirow{2}{*}{$\text{LRT}_i$} & Single Trial Analysis & \multirow{2}{*}{VSOM REML}  &  \multirow{2}{*}{No} \\		
		& MET Analysis &   &  \\				
		\multirow{2}{*}{$\text{GCD}_i$} & Single Trial Analysis & \multirow{2}{*}{one-step updated} &  \multirow{2}{*}{Intermediate} \\		
		& MET Analysis &  &  \\		
		\bottomrule 
	\end{tabular}
	\caption{The list of statistic for outlier detection presented in this article. Here $t_i$ is the studentised conditional residual (equivalent to mean shift effect of MSOM with variance fixed from the baseline model); $s_i^2$ is the Wald test statistic of mean shift effect from MSOM; $a\text{LRT}_i$ is the analytically derived log-likelihood ratio statistic of comparing VSOM to baseline (assuming that the variance are fixed from baseline model); $\text{LRT}_i$ is the true log-likelihood ratio test statistic; and $\text{GCD}_i$ is the generalised Cook's distance. Each statistic can be computed based on a single trial or MET model as a baseline model. Clearly where the variance estimation was fixed from baseline model, the statistic is more easily computed.}\label{tab:method}	
\end{table}
\end{center}

\subsubsection{Generalised Cook's Distance (GCD)}\label{sec:cook}
\cite{Cook1977} proposed one of the most popular measures of influence for linear fixed models with independent and identically distributed (i.i.d.) errors. The i.i.d. assumption of the so-called Cook's distance can be relaxed with a generalised Cook's distance for the fixed effects given as used by \cite{Christensen1992}:
\begin{equation}
GCD_i = (\hat{\bs{\theta}} - \hat{\bs{\theta}}_{[i]})^\top \hat{\text{cov}}(\hat{\bs{\theta}})^{-1}  (\hat{\bs{\theta}} - \hat{\bs{\theta}}_{[i]}) \label{eq:gcd}
\end{equation}
where $\hat{\bs{\theta}}$ is the estimated parameter vector of interest in a model and $\hat{\bs{\theta}}_{[i]}$ is estimated parameter vector with the $i$-th response $y_i$ deleted from the full observation $\bs{y}$.

In this article, we use $\bs{\theta} = \bs{\kappa}$ and $\hat{\text{cov}}(\hat{\bs{\kappa}})^{-1}$ is estimated from the average information matrix $\mathcal{I}_A$ \citep{Johnson1995} of the fit from the baseline model \eqref{eq:baseline} and $\hat{\bs{\kappa}}_{[i]}$ is the estimated variance parameters with the $y_i$ deleted from the full observation $\bs{y}$ with one-step update using the fit of the baseline model as initial values. That is, $\hat{\bs{\kappa}}_{[i]} = \hat{\bs{\kappa}} +  \left[\mathcal{I}_A\right]^{-1}\bs{v}(\hat{\bs{\kappa}})$ where $\bs{v}=\partial \ell_R / \partial \bs{\kappa}$ is the score function and $\ell_R$ is the log of the residual likelihood. Note if an element of $\bs{\kappa}$ was close to the boundary, then it was omitted from $\bs{\kappa}$. 

\subsubsection{Mean Shift Outlier Model (MSOM)}\label{sec:msom}

Suppose we consider one ($i$-th) observation at a time, the mean shift outlier model (MSOM) is given by adding an extra fixed effect to the baseline model \eqref{eq:baseline} as below.
\begin{equation}
\bs{y} = \bs{X}\bs{\tau} + \bs{Z}\bs{u}+ \bs{\delta}_i \phi_i + \bs{\epsilon} \label{eq:msom}
\end{equation}
where $\bs{\delta}_i$ is a $n\times 1$ binary vector where the $i$-th position is $1$ and $0$ elsewhere and $\phi_i$ is the fixed mean shift effect for the $i$-th observation. We can estimate $\phi_i$ by its E-BLUE, however, this requires fitting $n$ corresponding models and this is clearly computationally expensive.  However if the variance parameters are known, we can see from Proposition~\ref{prop:mme} in Appendix that $\hat{\phi}_i = r_i$ where  $r_i = (\bs{\delta}_i^\tp\bs{P}_X\bs{\delta}_i)^{-1}\bs{\delta}_i^\tp\bs{P}_X\bs{y}$ is the conditional residual for the $i$-th observation. Furthermore, we have $\text{var}(r_i) = (\bs{\delta}_i^\tp\bs{P}_X\bs{\delta}_i)^{-1}$ and as such a studentised conditional residual may be obtained $t_i = r_i  / \sqrt{\text{var}(r_i) }$ as a test-statistic for $H_0: \phi_i = 0$. The studentised conditional residuals are estimated by fit of the baseline model \eqref{eq:baseline} alone and so is computationally efficient. 

\subsubsection{Variance Shift Outlier Model (VSOM)}
The variance shift outlier model (VSOM) is similar to MSOM except that the extra effect is a random effect, i.e.,  
\begin{equation}
\bs{y} = \bs{X}\bs{\tau} + \bs{Z}\bs{u}+ {\bs{\delta}_i o_i} + \bs{\epsilon} \label{eq:vsom}
\end{equation}
where $\bs{\delta}_i$ is a $n\times 1$ binary vector where the $i$-th position is $1$ and $0$ elsewhere and  $o_i \sim N(0, \omega_i \bs{\delta}_i^\tp\bs{R}\bs{\delta}_i )$ with $o_i$ is independent of $\bs{u}$ and $\bs{\epsilon}$. It can be easily seen that $\text{var}(y_i)$ of VSOM has a shifted variance of $\omega_i \bs{\delta}_i^\tp\bs{R}\bs{\delta}_i$ compared to the $i$-th observation in the baseline model \eqref{eq:baseline}. We note that, like the MSOM, fitting $n$ mixed models demands higher computing time. \cite{Bernal-Vasquez2016} also noted this and opted VSOM out of consideration, however, \cite{Gogel1997} and \cite{Gumedze2010} consider fixing the variance components from the baseline model and derived an analytical form of the log-likelihood ratio statistic for $H_0: \omega_i=0$ vs. $H_1: \omega_i > 0$. This analytic form is given as:
$${a\text{LRT}_i} = (n - p  - 1)\log\left(\frac{n -p - 1}{n - p - t_i^2}\right) - \log t_i^2$$
for fixing $\bs{\kappa} = \hat{\bs{\kappa}}$ and $t_i^2 > 1$. We emphasise that these statistics are derived from fit of one model, thus, is computationally efficient. Note that this log-likelihood ratio statistic is based on  residual likelihood:
$${\text{LRT}_i} = 2\left(\ell_R(\text{VSOM}) - \ell_R(\text{Baseline})\right) $$
where $\ell_R(\text{VSOM})$ and $\ell_R(\text{Baseline})$ are residual likelihood of the VSOM and Baseline model respectively.

\subsubsection{Calibration and adjustment for multiple testing of mean shift effect}\label{sec:qval}
We calculate the $p$-values of $H_0: \phi_i = 0$ vs. $H_1: \phi_i \neq 0$ by using Wald tests using:
\begin{enumerate}
	\item $P(\chi^2_1>t_i^2)$ if we estimated $\phi_i$ by $r_i$; or 
	\item $P(\chi^2_1 > s_i^2)$ if E-BLUE was used to estimate $\phi_i$. 
\end{enumerate}
We note that both $p$-value is using a misspecified distribution as the variance parameters are estimated, however, in the absence of computationally efficient method we forgo this - a similar feat was conducted by \cite{Bernal-Vasquez2016}. These $p$-values are then adjusted for multiple testing by using \cite{Holm1979} implemented in statistical software R \citep{R2008} as function \texttt{p.adjust} with \texttt{method=``holm''}. Note that if the baseline model was based on a single trial then the adjustment is made on a per trial basis. 

\subsection{Robust prediction}\label{sec:robust}
\subsubsection{Mean shift as substitute for deletion}
Proposition~\ref{thm:4} shows that the $\hat{\bs{\tau}}$ in MSOM is the same as deleting the $i$-th observation and fitting the baseline model if the variance components are known. {We have a similar result for the random effects since  $\tilde{\bs{u}} = \bs{G}\bs{Z}^\tp\bs{V}^{-1}(\bs{y} - \bs{X}\hat{\bs{\tau}} - \bs{\delta}_i\hat\phi_i)$. 
	
	As the $i$-th observation is completely indexed by $\bs{\delta}_i$, naturally the corresponding marginal residual $\hat{\epsilon}_i = y_i - \bs{x}_i^\tp\hat{\bs{\tau}} - \hat\phi_i$ where $\bs{x}_i$ is the $i$-th row of $\bs{X}$ would be 0. If a particular random effect, $u_k$, is associated only with the $i$-th observation then in this case $\tilde{u}_k=0$. This will be often the case for trials with $p$-rep designs where the G$\times$E effect may completely regress to the mean.  }

\subsubsection{Variance shift for down-weighting}
For VSOM, the $i$-th observation has a larger (shifted) variance of $\omega_i \bs{\delta}_i^\tp\bs{R}\bs{\delta}_i \geq 0$ compared to the baseline model.  For a known $\bs{V}$, the BLUE of $\bs{\tau}$ is equivalent to the solution of a weighted least squares $\hat{\bs{\tau}} = (\bs{X}^\tp\bs{W}\bs{X})^-\bs{X}^\tp\bs{W}\bs{y}$ where $\bs{W} = \bs{V}^{-1}$. It is easy to see that in light of a larger variance, the weight of the $i$-th observation is smaller and thus down-weighted for the estimation of $\bs{\tau}$. {The down-weighting of $i$-th observation can also be seen occur for random effects as $\tilde{\bs{u}} = \bs{G}\bs{Z}^\tp\bs{V}^{-1}(\bs{y} - \bs{X}\hat{\bs{\tau}})$ and so with a large $i$-th diagonal element of $\bs{V}$, the corresponding marginal residual will contribute less to the prediction of $\bs{u}$.}

\subsubsection{Simple robust prediction}
The two aforementioned models introduce an easily applicable robust modelling where we fit a model where each observations that are identified as outliers are fitted with separate mean or variance shift effects. In this article, MSOM and VSOM conducted for robust prediction are conducted using a baseline MET model with each observations, that is flagged as an outlier, fitted as a separate mean/variance shifted effect. More explicitly, if a set of observations $\mathcal{O}$ are identified as outliers then we fit the model
$$\bs{y} = \bs{X}\bs{\tau} + \bs{Z}\bs{u} + \sum_{i \in  \mathcal{O}} \bs{\delta}_i\phi_i + \bs{\epsilon}$$
where $\phi_i$s are fixed effects for MSOM and random effects for VSOM. 

In this article, an observation is flagged as an outlier based on the adjusted $p$-value of $t_i$ from the baseline MET model using a threshold of 0.05. 

\section{Simulation}\label{sec:sim}

We construct a total of 3000 simulated data based on three different settings (1000 simulations for each setting). Setting 1 is based on an alpha-lattice design that comprises 27 outliers out of 700 observations while setting 2 and 3 are based on a $p$-rep design that comprises 27 and 174 outliers out of 2131 observations, respectively. More specifically, in setting 1, we simulate the data from a parametric bootstrap from the fitted MET model for ESWYT (Table~\ref{tab:var} and \ref{tab:vargxe}). We perturb the simulated data to introduce outliers in the simulated data  as follows. We randomly select three trials, $j_1$, $j_2$ and $j_3$. In the first trial we randomly select $3$ plots and contaminate one-third of these plots by adding noise sampled from $N(4\sigma_{j_1}, \sigma_{j_1}^2)$, $N(7\sigma_{j_1}, \sigma_{j_1}^2)$ and $N(10\sigma_{j_1}, \sigma_{j_1}^2)$ where $\sigma_{j_1}^2$ correspond to the error variance for the $j_1$-th trial. The same is repeated for the second and third trials except we randomly select $9$ and $15$ plots with noise variance replaced with $\sigma_{j_2}$ and $\sigma_{j_3}$, respectively. In setting 2, we repeat the same as before except using the fitted MET model for CAIGE. In setting 3, we repeat the same as the setting 2 except we perturb a higher number of plots by randomly selecting 9, 45, and 120 plots instead of 3, 9, and 15 plots.  

We fit the data generated model to the simulated data and for each observation we estimate the statistics outlined in Table~\ref{tab:method}. Subsequently, observations that have an adjusted $p$-value $< 0.05$ based on $t_i$ using MET analysis as baseline model (see Section~\ref{sec:qval}) are flagged as outliers then we fit the four corresponding models:
\begin{enumerate}[A)]
	\item Fit the non-contaminated data-generated model  (Baseline Model). 
	\item Delete observations flagged as outliers and fit the data-generated model (Deletion Model)
	\item Fit those identified as outliers with a separate mean shifted effect (MSOM).
	\item Fit those identified as outliers with a separate variance shifted effect (VSOM).
\end{enumerate}		
Clearly in the above models A) is not robust. For each simulation, we evaluate its effectiveness based on the simulation-based accuracy 
\begin{equation}
\text{cor}(\tilde{\bs{u}}_{ge}, \bs{u}_{ge}) \label{eq:acc}
\end{equation}
where $\tilde{\bs{u}}_{ge}$ is the E-BLUP under Baseline Model, Deletion Model, MSOM or VSOM.

\section{Results}

\subsection{Outlier discrimination}\label{sec:disc}
A problem that is coupled with choosing a statistic for outlier detection is the issue of choosing the threshold for outlier classification. To circumvent this issue temporary, we assess the performance of the statistic for outlier detection by looking at how well it can discriminate between observations labelled as outliers (observations with added noise) and non-outliers (observations with no noise added). Large values of all the statistics presented in Table~\ref{tab:method} are suggestive of an outlying observation. An ideal statistic will have larger values for outlying observations compared to non-outlying observations. 

To assess the performance of outlier discrimination, we can use the Wilcoxon rank-sum test statistic \citep{Mann1947} or equivalently (and possibly more familiarly) to the area under the receiver operating characteristic (ROC) curve \citep[][]{Hanley1982}. An ideal classifier will give an area under the ROC (aROC) of 1 while a random classifier will give an aROC of 0.5. 

The results of the aROC (Table~\ref{tab:auc}) indicate that MSOM perform best using MET analysis as the baseline model for all three simulation settings. There is little difference between using the computational efficient $t_i^2$ and $s_i^2$. 
\begin{center}
\begin{table}[h!]
	\centering
\begin{adjustwidth}{-1.5in}{0in} 
	\begin{tabular}{cccc}
		\toprule
		Simulation Setting & Statistic & Single Trial Analysis & MET Analysis \\ \midrule
		CAIGE & $t_i^2$ & (0.601, 0.663, 0.682, 0.714, 0.772) & (0.625, 0.677, 0.700, 0.733, 0.785) \\ 
		(higher outlier numbers) & $s_i^2$ & (0.603, 0.663, 0.682, 0.713, 0.771) & (0.628, 0.678, 0.700, 0.732, 0.783) \\ 
		& $a\text{LRT}_i$ & (0.458, 0.518, 0.535, 0.550, 0.605) & (0.460, 0.531, 0.545, 0.560, 0.602) \\ 
		& $\text{LRT}_i$ & (0.571, 0.655, 0.677, 0.701, 0.781) & (0.466, 0.626, 0.663, 0.696, 0.780) \\ 
		& $\text{GCD}_i$ & (0.598, 0.736, 0.774, 0.811, 0.905) & (0.461, 0.581, 0.666, 0.702, 0.760) \\ \midrule
		CAIGE & $t_i^2$ & (0.757, 0.880, 0.906, 0.929, 0.986) & (0.797, 0.905, 0.929, 0.950, 0.992) \\ 
		(lower outlier numbers) & $s_i^2$ & (0.758, 0.880, 0.906, 0.928, 0.985) & (0.795, 0.905, 0.929, 0.950, 0.991) \\ 
		& $a\text{LRT}_i$ & (0.590, 0.742, 0.776, 0.811, 0.901) & (0.647, 0.781, 0.813, 0.845, 0.931) \\ 
		& $\text{LRT}_i$ & (0.700, 0.858, 0.889, 0.918, 0.986) & (0.743, 0.879, 0.912, 0.943, 0.992) \\ 
		& $\text{GCD}_i$ & (0.722, 0.845, 0.877, 0.906, 0.987) & (0.593, 0.893, 0.918, 0.940, 0.985) \\  \midrule
		ESWYT & $t_i^2$ & (0.707, 0.853, 0.875, 0.895, 0.961) & (0.767, 0.860, 0.881, 0.901, 0.960) \\ 
		& $s_i^2$ & (0.709, 0.853, 0.874, 0.893, 0.957) & (0.764, 0.860, 0.882, 0.900, 0.957) \\ 
		& $a\text{LRT}_i$ & (0.395, 0.658, 0.683, 0.712, 0.818) & (0.413, 0.662, 0.687, 0.713, 0.812) \\ 
		& $\text{LRT}_i$ & (0.705, 0.821, 0.848, 0.875, 0.955) & (0.718, 0.838, 0.862, 0.885, 0.953) \\ 
		& $\text{GCD}_i$ & (0.693, 0.821, 0.860, 0.894, 0.970) & (0.420, 0.713, 0.784, 0.850, 0.955) \\ 
		\bottomrule
	\end{tabular}
	\caption{The five number of summary of the 1000 aROCs of different statistics using a single trial or MET analysis by the three simulating setting (see Section~\ref{sec:sim} for more information). The list of statistic is shown in Table~\ref{tab:method}. }
	\label{tab:auc}
\end{adjustwidth}
\end{table}
\end{center}

\subsection{Computational efficiency vs. accuracy}\label{sec:prec}
We would expect that fixing the variance parameters at the baseline loses in the discrimination power compared to refitting the model -- this is indeed true for VSOM with a significant gain in aROC however the difference is minimal for MSOM (Table~\ref{tab:auc}). An attractive feature of fixing variance parameters from the baseline model is the computational efficiency which is important from a practical aspect. As MSOM appears to work well in computational aspect and offer the best discrimination, we explore further by examining it's outlier classification ability based on the adjusted $p$-value (see Section~\ref{sec:qval}). We use a adjusted $p$-value threshold of 0.05 to classify as an outlier.  To assess the performance of outlier classification, we use precision (the fraction of true positives over all positives), recall (the fraction of true positive over true positive and false negatives, also known as sensitivity) and the F1 score (the harmonic average of precision and recall). Ideally you will have a classify with precision and recall of 1, however, often classifiers that outperform in precision, do poorer in recall and vice versa. F1 score  is a combination of the precision and recall however depending on the objective, a score that weighs more on say, recall may be desirable. For example, in the context of outlier detection for the purpose of re-examining the observations, it may be desirable to have less false positive (higher precision) if the cost of re-examination is expensive. For our simple robust genomic prediction (Section~\ref{sec:robust}), it may less critical to have false positives and higher recall may be favoured. We see in general that we lose precision but gain in recall by using $s_i^2$ over $t_i^2$ (Table~\ref{tab:recall}). Furthermore, F1 scores are favourable for $s_i^2$ over $t_i^2$ for all three simulation settings. 

\begin{table}[ht]
\begin{adjustwidth}{-1in}{0in} 

	\centering
	\begin{tabular}{cccccc}
		\toprule
		Simulation Setting & Adjusted $p$-value & Baseline model & Precision & Recall & F1 \\ 
		\midrule
		CAIGE & \multirow{2}{*}{$t_i^2$} & Single Site Analysis & 0.878 & 0.018 & 0.035 \\ 
		(higher outlier numbers) &  & MET Analysis & 0.982 & 0.015 & 0.031 \\ 
		& \multirow{2}{*}{$s_i^2$} & Single Site Analysis & 0.767 & 0.022 & 0.042 \\ 
		&  & MET Analysis & 0.911 & 0.019 & 0.037 \\ \midrule
		CAIGE & \multirow{2}{*}{$t_i^2$} & Single Site Analysis & 0.869 & 0.235 & 0.365 \\ 
		(lower outlier numbers) &  & MET Analysis & 0.981 & 0.224 & 0.359 \\ 
		& \multirow{2}{*}{$s_i^2$} & Single Site Analysis & 0.763 & 0.276 & 0.401 \\ 
		&  & MET Analysis & 0.911 & 0.274 & 0.416 \\ \midrule
		ESWYT & \multirow{2}{*}{$t_i^2$} & Single Site Analysis & 0.955 & 0.133 & 0.230 \\ 
		&  & MET Analysis & 0.995 & 0.076 & 0.142 \\ 
		& \multirow{2}{*}{$s_i^2$} & Single Site Analysis & 0.801 & 0.185 & 0.298 \\ 
		&  & MET Analysis & 0.893 & 0.131 & 0.227 \\ 
		\bottomrule
	\end{tabular}
	\caption{The above table shows the average precision, recall, and F1 score across the 1000 simulations for classification of outliers using adjusted $p$-value for $t_i^2$ and $s_i^2$ (Section~\ref{sec:qval}).} \label{tab:recall}
\end{adjustwidth}
\end{table}

\subsection{Outlier detection via single trial or MET analysis?}\label{sec:cases}
Another pending question is whether we should conduct outlier detection based on per trial analysis or a combined MET analysis. We see in Table~\ref{tab:recall} that precision increases however recall decreases if we use a baseline model based on MET analysis over single trial analysis. F1 scores are generally higher for the single site analysis, however, as discussed in Section~\ref{sec:prec}, depending on the objective it may be desirable to use a score with different weights for precision and recall. We illustrate this with example application to the real data CAIGE next. 

The observed yield distribution for ESWYT and CAIGE is shown in Figure~\ref{fig:boxplot} marked with the identified outliers according to adjusted $p$-value of $t_i^2$ based on either single site analysis or MET analysis and $s_i^2$ based MET analysis. Single trial analysis and MET analysis differ in that the MET analysis borrows the strength across trials. 

Case 1 in Figure~\ref{fig:cases} show that the adjusted $p$-value based on MET analysis identified the yield of genotype G35 at Horsham as an outlier where as the single site analysis did not.  Yield of genotype G35 is one of the lowest observed in Horsham however examination of the yield of genotype G35 in other trials indicate a medium to above average performance. 

A similar observation as Case 1 is seen in Case 2 except the outlier is identified only by $s_i^2$ of the MET analysis, perhaps attesting to the higher precision of $s_i^2$ observed in the simulations (Table~\ref{tab:recall}). 

In Case 3, we see that the $t_i$ of single trial analysis identifies the genotype G6922234 at Toodyay as an outlier while MET analysis based mean shift did not. We can see that the flagged outlier in Toodyay is the smallest yield observed however the performance of the same genotype across trials indicate that this genotype is consistently low performing and perhaps not particularly unusual. 

Finally for Case 4, we observe a cautionary tale for outlier detection. Two observations in Balaklava that are flagged as outliers by all three adjusted $p$-values belong to Genotype 26. The reason for this can easily be seen from large variance between the two observations. Naturally the prediction of genotype G26 at Balaklava resides in between these two observations with these two observations result with large conditional residuals. In this case both observations are flagged as outliers however removal of either one of the observation will likely result in the other not being flagged as an outlier. 

\begin{center}
\begin{figure}[ht!]
	\includegraphics[width=0.49\textwidth]{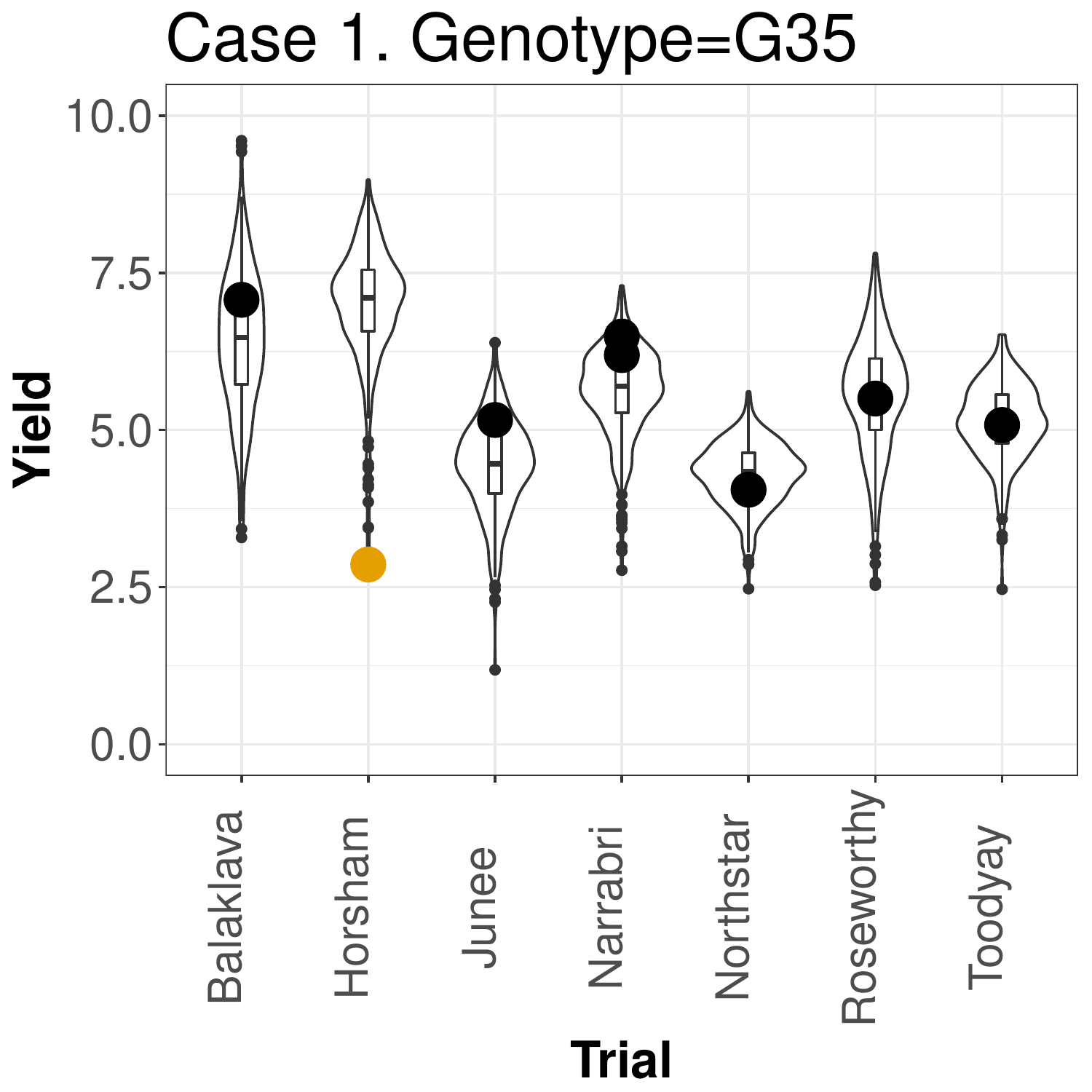}
	\includegraphics[width=0.49\textwidth]{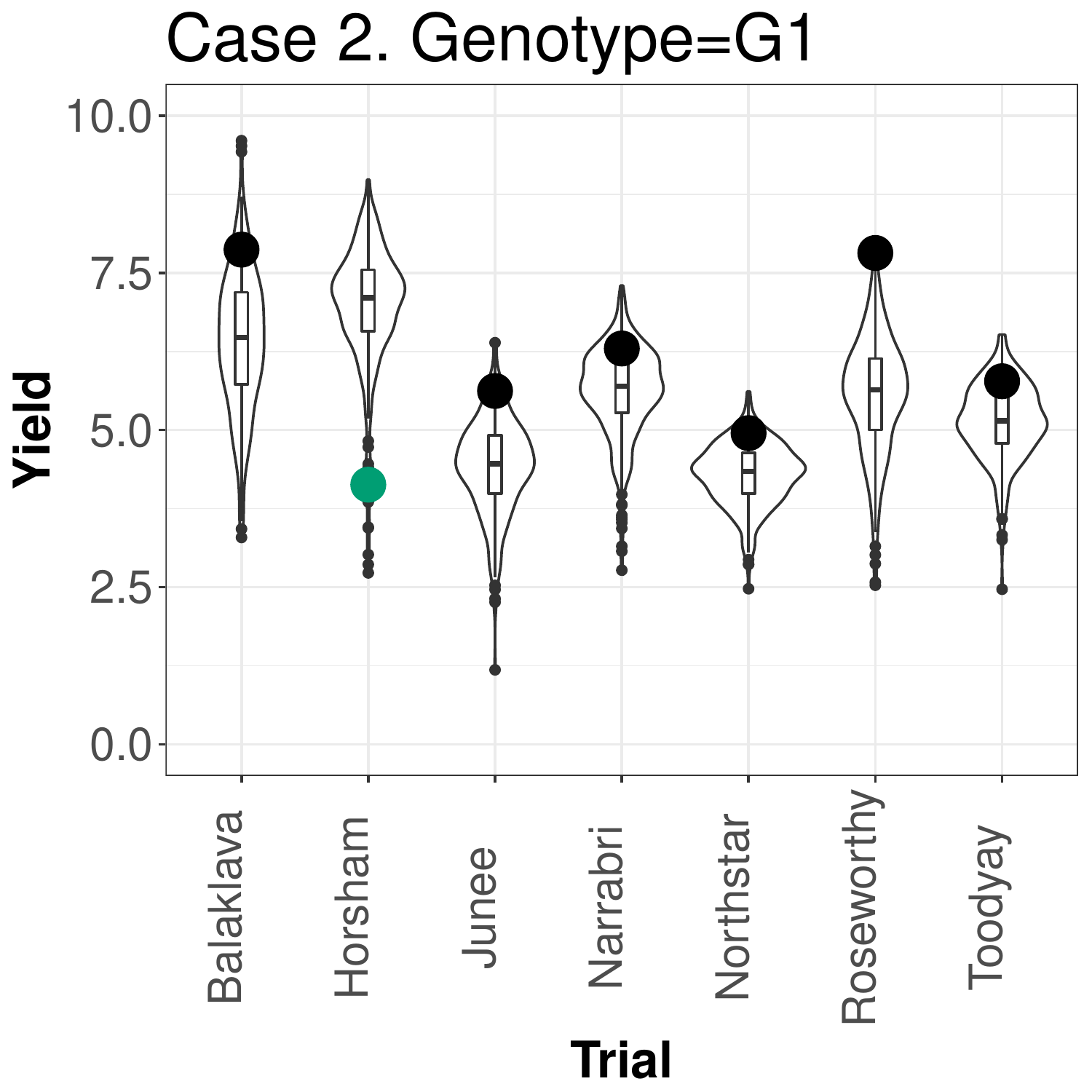}
	\includegraphics[width=0.49\textwidth]{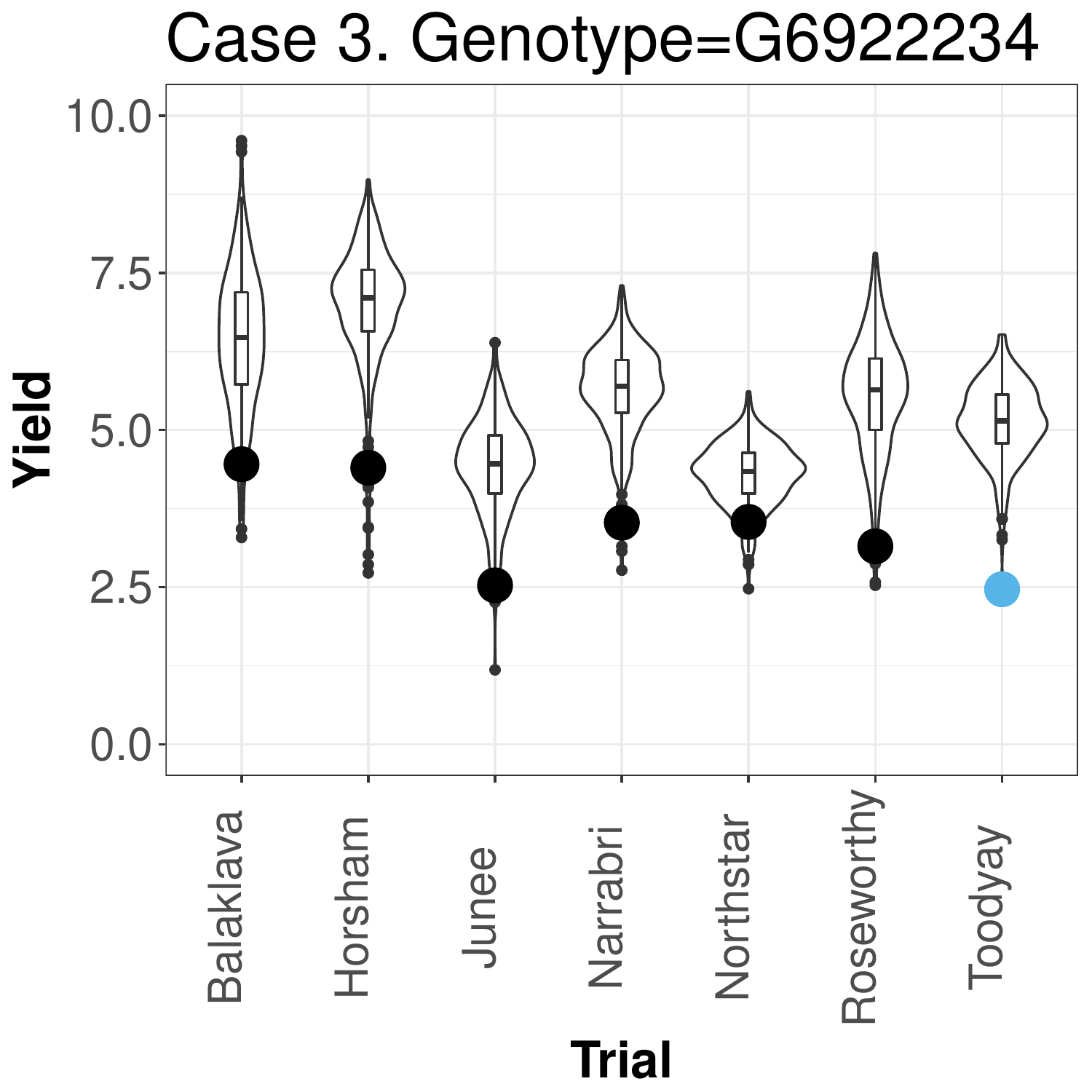}
	\includegraphics[width=0.49\textwidth]{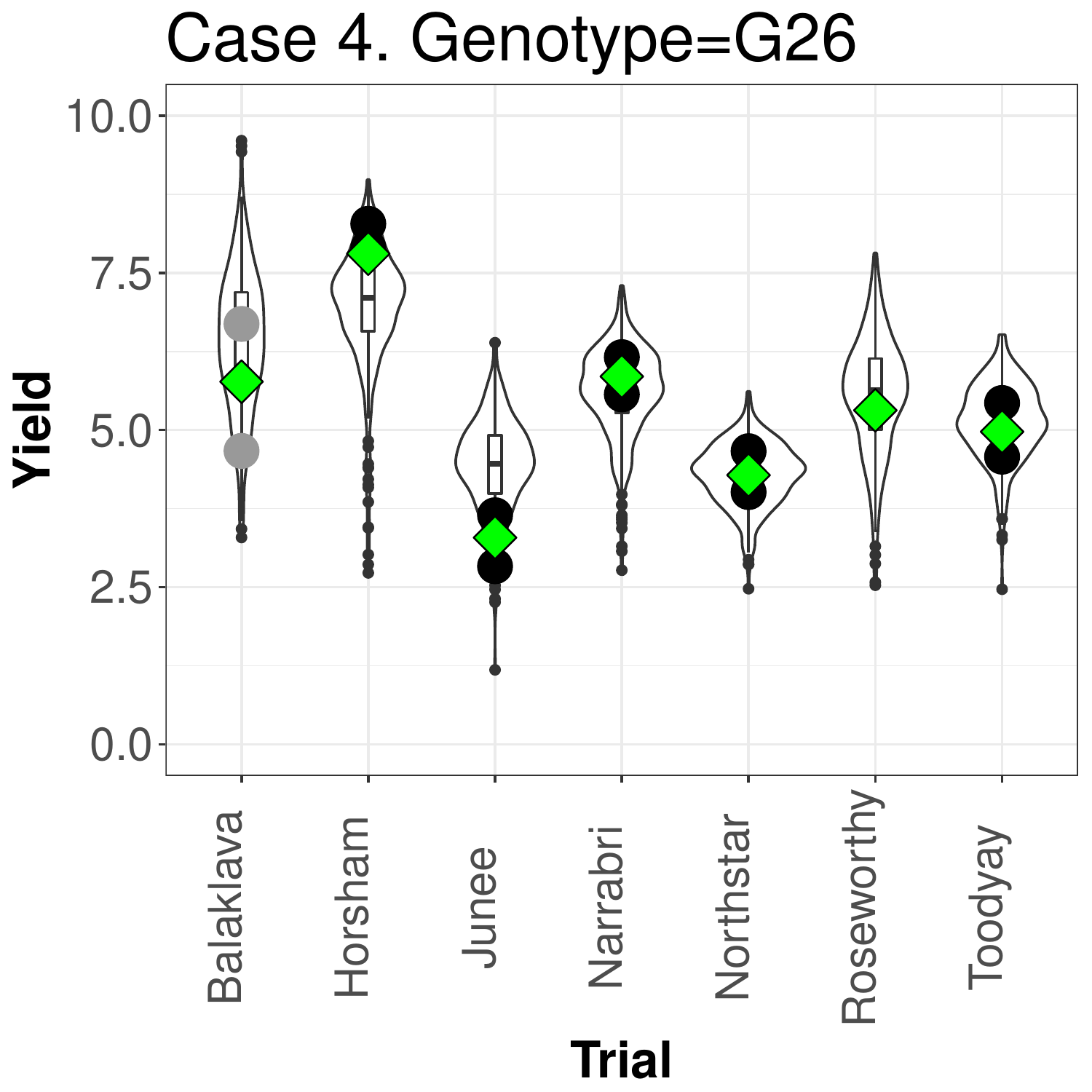}
	\caption{Above figures show the yield of four genotypes (G26, G35, G6922234, and G1) in CAIGE as enlarged circle points overlaid on the observed yield distribution from Figure~\ref{fig:boxplot}. The grey, yellow, blue and green circles indicate observations flagged as outliers per the adjusted $p$-value in Figure~\ref{fig:boxplot} and the green diamond in Case 1 indicate the predicted value of G26 at the corresponding trial under the baseline MET model. Each cases are discussed in-depth in Section~\ref{sec:cases}. }\label{fig:cases}
\end{figure}
\end{center}

\subsection{Robust genomic prediction}

Another use of MSOM and VSOM is that we can consider using it as a simple robust model (Section~\ref{sec:robust}). Table~\ref{tab:acc} shows that the relative gain of G$\times$E accuracy for using MSOM has a roughly similar median regardless of the number of identified outliers, however, we can see in Figure~\ref{fig:gxeaccuracy} B) that if all the correct outliers are identified and used for the MSOM then virtually all simulations would have had a higher accuracy. In practice, of course we are likely to misidentify or overlook real outliers and so the gain see in Figure~\ref{fig:gxeaccuracy} A) is more realistic. Figure~\ref{fig:gxeaccuracy} A) also shows that most simulated data gain in G$\times$E accuracy using a MET MSOM (78.6-91.2\%, see Table~\ref{tab:acc}), although the gain seems to reduce if there are higher number of outliers in the data owing likely to difficulty in identifying outliers due to swamping.

The outlying observations in our simulated data was generated from a shifted mean so we would expect under this situation that the MSOM will be a better fitting model. This indeed is reflected in the accuracy of the G$\times$E prediction of MSOM vs VSOM as seen in Figure~\ref{fig:gxeaccuracy} C). However, it should be noted that Figure~\ref{fig:gxeaccuracy} C) is using an ideal case with known true outliers and it appears that in a more realistic case where the outliers were identified by $t_i^2$ for a single trial analysis, VSOM has a similar gain in accuracy as MSOM. 

\begin{table}[ht]
	\centering
	\begin{tabular}{ccccc}
		\toprule
		Number of & \% relative accuracy gain $p_i$  & \multicolumn{2}{c}{\% $p_i \ge 0$} & 	Number of\\		 \cmidrule{3-4}
		identified outliers&  MSOM $-$ Baseline& MSOM & VSOM & simulations\\ 
		\midrule  \multicolumn{2}{l}{~~~~CAIGE (higher outlier numbers)} & & &\\ \midrule 
		0 & (0.0, 0.0, 0.0, 0.0, 0.0) & 100.0 & 100.0 &  69 \\ 
		1 & (-0.8,  0.0,  0.1,  0.2,  0.9) & 69.4 & 69.4 & 160 \\ 
		2 & (-0.7,  0.0,  0.2,  0.3,  2.1) & 77.4 & 77.9 & 226 \\ 
		3 & (-1.3,  0.1,  0.2,  0.4,  1.7) & 82.5 & 84.3 & 280 \\ 
		4 & (-1.3,  0.1,  0.3,  0.5,  3.1) & 82.0 & 81.0 & 189 \\ 
		5 & (-0.6,  0.0,  0.2,  0.5,  2.5) & 74.6 & 76.3 &  59 \\ 
		6 & (-0.3,  0.0,  0.3,  0.5,  1.3) & 80.0 & 80.0 &  15 \\ 
		7 & (0.2, 0.2, 0.4, 0.7, 0.7) & 100.0 & 100.0 &   2 \\
		All &(-1.3,  0.0,  0.2,  0.3,  3.1)&  79.9 & 80.4 & 1000 \\ 
		\midrule  \multicolumn{2}{l}{~~~~CAIGE (lower outlier numbers)} &  & & \\ \midrule  
		1 & (0.0, 0.1, 0.2, 0.3, 0.8) & 90.0 & 90.0 &  10 \\ 
		2 & (-0.5,  0.1,  0.2,  0.3,  0.6) & 87.0 & 87.0 &  23 \\ 
		3 & (-0.3,  0.1,  0.2,  0.4,  1.0) & 87.3 & 89.1 &  55 \\ 
		4 & (-0.3,  0.1,  0.3,  0.6,  3.6) & 88.3 & 89.3 & 103 \\ 
		5 & (-1.0,  0.2,  0.4,  0.6,  3.0) & 89.2 & 89.8 & 176 \\ 
		6 & (-2.5,  0.2,  0.4,  0.6,  2.4) & 90.7 & 91.2 & 205 \\ 
		7 & (-0.6,  0.3,  0.5,  0.8,  2.0) & 92.7 & 93.2 & 177 \\ 
		8 & (-0.5,  0.3,  0.5,  0.8,  3.0) & 94.1 & 94.1 & 136 \\ 
		9 & (-0.3,  0.3,  0.5,  0.8,  2.6) & 98.5 & 98.5 &  67 \\ 
		10 & (-1.4,  0.3,  0.7,  0.9,  1.9) & 88.6 & 88.6 &  35 \\ 
		11 & (-0.1,  0.3,  0.6,  0.8,  1.9) & 92.3 & 100.0 &  13 \\ 
		All & (-2.5,  0.2,  0.4,  0.7,  3.6) & 91.2 & 91.8 &  1000\\
		\midrule \multicolumn{2}{l}{~~~~ESWYT}  & & &  \\  \midrule
		0 & (0.0, 0.0, 0.0, 0.0, 0.0) & 100.0 & 100.0 &  18 \\ 
		1 & (-2.5,  0.0,  0.5,  1.1,  4.6) & 78.3 & 77.9 & 253 \\ 
		2 & (-2.2,  0.1,  0.6,  1.3,  6.4) & 77.7 & 78.7 & 431 \\ 
		3 & (-2.4,  0.1,  0.7,  1.8,  5.5) & 78.5 & 79.3 & 246 \\ 
		4 & (-2.5,  0.1,  0.8,  2.2,  5.8) & 80.0 & 80.0 &  50 \\ 
		5 & (0.1, 0.1, 1.5, 3.0, 3.0) & 100.0 & 100.0 &   2 \\
		All & (-2.5,  0.0,  0.6,  1.4,  6.4) & 78.6 & 79.1  & 1000\\
		\bottomrule
	\end{tabular}
	\caption{The first column shows the number of outliers identified using adjusted $p$-value of $t_i^2$ of the single trial analysis with the subsequent columns displaying the five number summary of the percentage relative gain in G$\times$E accuracy using MSOM MET model; the percentage of simulations where G$\times$E accuracy was greater or equal for VSOM and MSOM compared to baseline MET model; and the number of simulations associated with the number of identified outliers.}\label{tab:acc}
\end{table}

\begin{figure}[ht!]
	\centering\includegraphics[width=\textwidth]{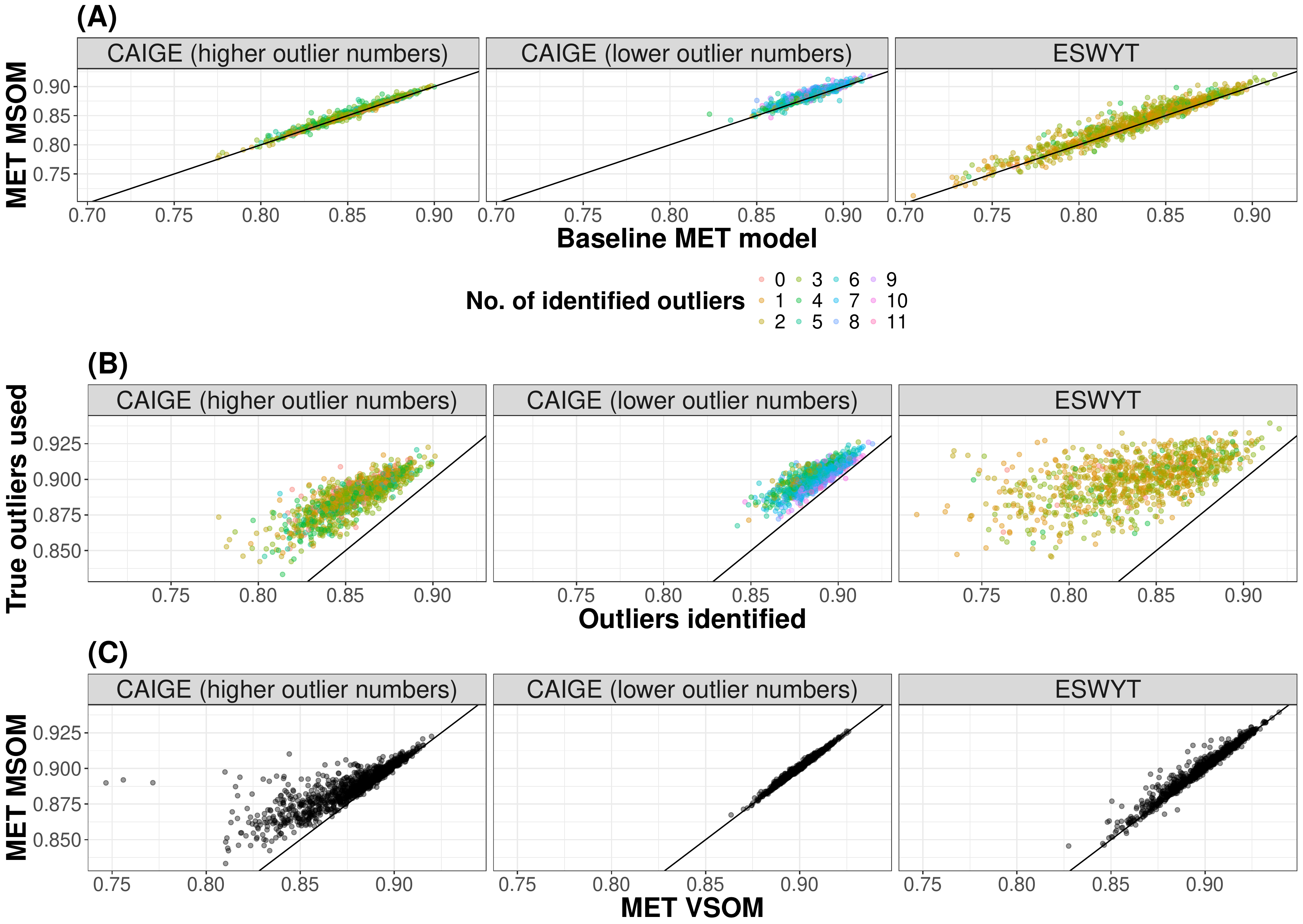}
	\caption{The figures above show the scatterplot of the accuracy of G$\times$E effects \eqref{eq:acc} by simulation setting for A) baseline MET model vs. MET MSOM where  the observations that were flagged as an outlier according to the adjusted $p$-value of $t_i^2$ of single trial analysis; B) MET MSOM of where outliers were identified by the adjusted $p$-value of $t_i^2$ of single trial analysis vs. MET MSOM where the true outlier labels were used; and C) MET VSOM vs. MET MSOM where both models used outliers from the true label. Each point correspond to the accuracy of G$\times$E for a particular simulated data and the black line corresponds to $y=x$. If the point is above the $y=x$ line, the model on the vertical axis is performing better for prediction of G$\times$E effects. }\label{fig:gxeaccuracy}
\end{figure}

\section{Discussion}

In this paper, we have shown that the studentised conditional residuals $t_i^2$ from a MET model offer a higher discrimination power than using a single trial analysis (Section~\ref{sec:disc}). We show this increases also in precision however with some decrease in recall (Table~\ref{tab:recall}). Depending on the objective, a higher precision may be more desirable and we illustrate particular cases on the real CAIGE data (Section~\ref{sec:cases}) of where such advantages may occur by using a MET model that borrows strength across trials. Borrowing strength may be desirable in particular for $p$-rep designs as there is only one plot of a genotype in a particular trial.  It should be noted that our MET data generating model comprised of positive genetic correlation between trials (Table~\ref{tab:vargxe}) and that borrowing strength may not necessary yield in better outlier discrimination performance.   

We further show that the mean shift effect is the conditional residual with variance fixed from the baseline model (Section~\ref{sec:msom}) and this follows that the calculation of $t_i$ is computational efficient requiring fit of the baseline model alone. A concern of course is that fixing the variance from baseline may compromise in the loss of power in outlier detection, however, our simulation results show that $t_i^2$ closely follows the discrimination power of $s_i^2$ (Table~\ref{tab:auc}) with gain in precision but loss in recall (Table~\ref{tab:auc}). As mention in Section~\ref{sec:prec}, it will depend on the objective whether a higher precision or a higher recall is favoured, however, for $n$ number of observations, we require fitting $n$ models for $s_i^2$ and thus from a practical aspect, $t_i^2$ would be preferred even with the loss of recall.

Additionally, we show that the MSOM and VSOM can serve as a simple robust model for genomic prediction (Section~\ref{sec:robust}). The benefit of this robust modelling is that MSOM and VSOM are simple addition to the baseline linear mixed model as such this poses little difficulty to apply in user preferred linear mixed model software. A more sophisticated fitting procedure of robust linear mixed models exist \citep{Koller2013, Lourenco2017} however these may yet lack features that the practitioner requires. 

Finally, MSOM and VSOM for robust prediction require a priori identification of outliers and a correct identification can reap more benefit in G$\times$E accuracy (Figure~\ref{fig:gxeaccuracy}). The identification of outliers is however a difficult task for a complex data such as this and the methods proposed (Table~\ref{tab:method}) consider only one observation at a time - a feat that results in difficultly to tackle swamping and masking. Future research will benefit with better outlier detection methods in linear mixed models that is practical for application and  user-friendly software development of more sophisticated robust linear mixed model fitting procedures.

\appendix

\section*{Appendix}

\begin{table}[ht]
	\centering
	\begin{tabular}{lrrrrr}
		\toprule
		ESWYT & & & & &\\ \toprule Trial & Trial Mean & Replicate & Sub-Block &  Residual & \\ \midrule
		L127 & 3.95 & 5.16E-08 & 3.11E-02 & 1.89E-01 &  \\ 
		L130 & 3.79 & 5.16E-08 & 1.24E-07 & 3.55E-01 &  \\ 
		L14 & 5.00 & 5.16E-08 & 1.31E-01 & 1.67E-01 &  \\ 
		L51 & 5.68 & 8.16E-07 & 1.85E-02 & 1.45E-01 &  \\ 
		L62 & 4.48 & 5.16E-08 & 5.37E-02 & 1.65E-01 &  \\ 
		L68 & 3.69 & 5.16E-08 & 2.35E-03 & 4.23E-02 &  \\ 
		L79 & 4.04 & 5.16E-08 & 9.31E-08 & 2.88E-02 &  \\ 
		\toprule CAIGE & & & & & \\ \toprule Trial & Trial Mean & Block & Row & Column & Residual \\ \midrule
		Balaklava & 6.38 & 2.26E-03 & 2.74E-02 & 5.59E-02 & 8.42E-02 \\ 
		Horsham & 6.94 & 7.30E-08 & 3.26E-03 & 7.57E-02 & 8.95E-02 \\ 
		Junee & 4.40 & 7.30E-08 & 1.51E-02 & 7.10E-02 & 7.74E-02 \\ 
		Narrabri & 5.59 & 2.75E-03 & 9.99E-03 & 8.25E-03 & 6.64E-02 \\ 
		Northstar & 4.29 & 3.76E-02 & 7.30E-08 & 1.42E-02 & 5.68E-02 \\ 
		Roseworthy & 5.51 & 1.14E-02 & 1.00E-02 & 2.32E-02 & 2.32E-01 \\ 
		Toodyay & 5.11 & 7.30E-08 & 3.70E-03 & 4.03E-02 & 7.00E-02 \\ 
		\bottomrule
	\end{tabular}
	\caption{The estimated variance components and the (fixed) trial mean of the MET model fitted to ESWYT and CAIGE.}\label{tab:var}
\end{table}

\begin{table}[ht]
	\centering
	\begin{tabular}{rrrrrrrr}
		\toprule
		ESWYT & & & & & & & \\
		\toprule
		& L127 & L130 & L14 & L51 & L62 & L68 & L79 \\ 
		\midrule
		L127 & 0.31 & 0.06 & 0.01 & 0.23 & 0.16 & 0.15 & 0.11 \\ 
		L130 & 0.06 & 0.74 & 0.41 & 0.16 & 0.16 & 0.12 & 0.06 \\ 
		L14 & 0.01 & 0.41 & 0.51 & 0.10 & 0.14 & 0.12 & 0.05 \\ 
		L51 & 0.23 & 0.16 & 0.10 & 0.39 & 0.13 & 0.12 & 0.08 \\ 
		L62 & 0.16 & 0.16 & 0.14 & 0.13 & 0.44 & 0.09 & 0.05 \\ 
		L68 & 0.15 & 0.12 & 0.12 & 0.12 & 0.09 & 0.14 & 0.09 \\ 
		L79 & 0.11 & 0.06 & 0.05 & 0.08 & 0.05 & 0.09 & 0.11 \\ 
		\bottomrule
		CAIGE & & & & & & & \\
		\toprule
		& Balaklava & Horsham & Junee & Narrabri & Northstar & Roseworthy & Toodyay \\ 
		\midrule
		Balaklava & 1.07 & 0.14 & 0.40 & 0.38 & 0.27 & 0.60 & 0.35 \\ 
		Horsham & 0.14 & 0.86 & 0.21 & 0.26 & 0.15 & 0.15 & 0.22 \\ 
		Junee & 0.40 & 0.21 & 0.47 & 0.28 & 0.20 & 0.30 & 0.28 \\ 
		Narrabri & 0.38 & 0.26 & 0.28 & 0.45 & 0.21 & 0.35 & 0.27 \\ 
		Northstar & 0.27 & 0.15 & 0.20 & 0.21 & 0.16 & 0.20 & 0.17 \\ 
		Roseworthy & 0.60 & 0.15 & 0.30 & 0.35 & 0.20 & 0.58 & 0.30 \\ 
		Toodyay & 0.35 & 0.22 & 0.28 & 0.27 & 0.17 & 0.30 & 0.28 \\ 
		\bottomrule
	\end{tabular}
	\caption{The estimated variance parameter of $\bs{G}_e$ of the MET model fitted to ESWYT and CAIGE.}\label{tab:vargxe}
\end{table}

\section{Software}\label{sec:software}
All models in this paper were fitted using the ASReml-R package \citep{Butler2009} within the R statistical environment \citep{R2008} which uses the average information algorithm \citep{Gilmour1995} for residual maximum likelihood (REML) estimation for variance parameters. Once the REML estimates of the variance parameters are obtained, a solution of the mixed model equations (MME) is used to provide the empirical best linear unbiased estimates (E-BLUEs) of the fixed effects and empirical best linear unbiased predictions (E-BLUPs) of the random effects \citep{Gilmour2004}.

\section{MSOM proofs}
We can rewrite the MSOM \eqref{eq:msom} as 
$$\bs{y} = \bs{X}_\delta\bs{\tau}_\delta + \bs{Z}\bs{u} + \bs{e} $$
where $\bs{X}_\delta\ = \begin{bmatrix}
\bs{X} & \bs{\delta}_i 
\end{bmatrix}$ and $\bs{\tau}_\delta = (\bs{\tau}^\tp, \phi_i)^\tp$. The mixed model equation of  \citep{Henderson1949} is given as 
\begin{equation}
\begin{bmatrix}
\bs{X}^\tp\bs{R}^{-1}\bs{X} &  \bs{X}^\tp\bs{R}^{-1}\bs{\delta}_i &  \bs{X}^\tp\bs{R}^{-1}\bs{Z}\\
\bs{\delta}_i^\tp\bs{R}^{-1}\bs{X} & \bs{\delta}_i^\tp\bs{R}^{-1}\bs{\delta}_i &   \bs{\delta}_i^\tp\bs{R}^{-1}\bs{Z}\\
\bs{Z}^\tp\bs{R}^{-1}\bs{X} & \bs{Z}^\tp\bs{R}^{-1}\bs{\delta}_i & \bs{Z}^\tp\bs{R}^{-1}\bs{Z} + \bs{G}^{-1}\\
\end{bmatrix}
\begin{bmatrix}
\hat{\bs{\tau}}\\
\hat{\phi}_i\\
\tilde{\bs{u}}\\
\end{bmatrix}
=
\begin{bmatrix}
\bs{X}^\tp\bs{R}^{-1}\bs{y}\\
\bs{\delta}_i^\tp\bs{R}^{-1}\bs{y}\\
\bs{Z}^\tp\bs{R}^{-1}\bs{y}
\end{bmatrix}
\label{eq:mme}
\end{equation}

In this section, we assume that the variance $\bs{G}$ and $\bs{R}$ are known.
\begin{prop}\label{prop:mme}
	\begin{eqnarray*}
		\hat{\bs{\tau}}  &=& 	(\bs{X}^\tp\bs{P}_{\delta_i}\bs{X})^{-}\bs{X}^\tp\bs{P}_{\delta_i}\bs{y}\\
		\hat{\phi}_i &=& 	(\bs{\delta}_i^\tp\bs{P}_X\bs{\delta}_i)^{-1}\bs{\delta}_i^\tp\bs{P}_X\bs{y}
	\end{eqnarray*}
	where $\bs{P}_X = \bs{V}^{-1} - \bs{V}^{-1}\bs{X}(\bs{X}^\tp\bs{V}^{-1}\bs{X})^-\bs{X}^\tp\bs{V}^{-1}$ and $\bs{P}_{\delta_i} = \bs{V}^{-1} - \bs{V}^{-1}\bs{\delta}_i(\bs{\delta}_i^\tp\bs{V}^{-1}\bs{\delta}_i)^{-1}\bs{\delta}_i^\tp\bs{V}^{-1}$.
\end{prop}
\begin{proof}
	\begin{eqnarray*}
		\begin{bmatrix}
			\hat{\bs{\tau}}\\
			\hat{\phi}_i 
		\end{bmatrix} &=& (\bs{X}_\delta \bs{V}^{-1}\bs{X}_\delta)^-\bs{X}_\delta^\tp\bs{V}^{-1}\bs{y} \\
		&=& \left(\begin{bmatrix}
			\bs{X}^\tp \\
			\bs{\delta}^\tp
		\end{bmatrix}\bs{V}^{-1}\begin{bmatrix}
			\bs{X} & \bs{\delta} 
		\end{bmatrix}
		\right)^- \begin{bmatrix}
			\bs{X}^\tp \\
			\bs{\delta}^\tp
		\end{bmatrix} \bs{V}^{-1}\bs{y} \\
		&=&
		\begin{bmatrix}
			\bs{X}^\tp\bs{V}^{-1}\bs{X} &\bs{X}^\tp\bs{V}^{-1}\bs{\delta}_i  \\
			\bs{\delta}_i^\tp\bs{V}^{-1}\bs{X} &\bs{\delta}_i^\tp\bs{V}^{-1}\bs{\delta}_i    \\
		\end{bmatrix}^- \begin{bmatrix}
			\bs{X}^\tp \\
			\bs{\delta}_i^\tp 
		\end{bmatrix} \bs{V}^{-1} \bs{y} \\
		&=&
		\begin{bmatrix}
			(\bs{X}^\tp\bs{P}_{\delta_i}\bs{X})^{-}\bs{X}^\tp\bs{P}_{\delta_i}\bs{y}\\
			(\bs{\delta}_i^\tp\bs{P}_X\bs{\delta}_i)^{-1}\bs{\delta}_i^\tp\bs{P}_X\bs{y}\\
		\end{bmatrix}
	\end{eqnarray*}
	where the last step follows from the standard (generalised) inverse of a partitioned matrix.
\end{proof}

Without loss of generality, assume henceforth that $i=1$ and we partition the corresponding vector and matrices as follows 
$$\bs{y} = \begin{bmatrix}
y_i \\
\bs{y}_{[i]}
\end{bmatrix}, \quad  \bs{X}_\delta = \begin{bmatrix}
\bs{x}_i^\tp & 1 \\
\bs{X}_{[i]} & \bs{0} \\
\end{bmatrix}, \quad \bs{V} = 
\begin{bmatrix}
v_{ii} & \bs{v}_i^\tp \\
\bs{v}_i & \bs{V}_{[i]}
\end{bmatrix}.
$$
where $\bs{y}_{[i]}$ and $\bs{X}_{[i]}$ are $\bs{y}$ and $\bs{X}$ with the $i$-th row removed, respectively; $\bs{x}_i^\tp$ is the $i$-th row of $\bs{X}$; $\bs{V}_{[i]}$ is the matrix $\bs{V}$ with the $i$-th row and column removed;  $v_{ii}$ is the $i$-th diagonal element of $\bs{V}$ and $\bs{v}_i$ is the $i$-th column of $\bs{V}$ with the $i$-th element removed. 

Suppose also 
$$\bs{V}^{-1} = 
\begin{bmatrix}
v^{ii} & \bs{\lambda}_i^\tp \\
\bs{\lambda}_i & \bs{\Lambda}_{[i]}
\end{bmatrix}.$$

\begin{prop}\label{prop:viinv}
	$	\bs{V}_{[i]}^{-1} = \bs{\Lambda}_{[i]}  -  \bs{\lambda}_i\bs{\lambda}_i^\tp/v^{ii}$
\end{prop}
\begin{proof}
	We have $\bs{V}\bs{V}^{-1} = \bs{I}_n$ and so it follows that $\bs{v}_i\bs{\lambda}_i^\tp + \bs{V}_{[i]}\bs{\Lambda}_{[i]} = \bs{I}_{n-1}$ and $\bs{v}_i v^{ii} + \bs{V}_{[i]}\bs{\lambda}_i = \bs{0}$.
	See also \cite{Christensen1992}.
\end{proof}
\begin{prop}\label{prop:Pdelta}
	$	\bs{P}_{\delta_i} = \begin{bmatrix}
	0 & \bs{0}^\tp \\
	\bs{0} & \bs{V}^{-1}_{[i]}
	\end{bmatrix}$
\end{prop}
\begin{proof}
	Clearly, $(\bs{\delta}^\tp\bs{V}^{-1}\bs{\delta})^{-1} = 1/v^{ii}$ and $\delta_i\delta_i^\tp = \begin{bmatrix}1 & \bs{0}_{1 \times (n-1)} \\
	\bs{0}_{(n-1)\times 1} & \bs{0}_{(n-1)\times (n-1)}
	\end{bmatrix}$. Thus 
	\begin{eqnarray*}
		\bs{P}_{\delta_i} &=& \bs{V}^{-1} - \bs{V}^{-1}\bs{\delta}_i(\bs{\delta}_i^\tp\bs{V}^{-1}\bs{\delta}_i)^{-1}\bs{\delta}_i^\tp\bs{V}^{-1}\\
		&=&  \bs{V}^{-1}  - \bs{V}^{-1}\bs{\delta}_i\bs{\delta}_i^\tp\bs{V}^{-1}/v^{ii} \\
		&=& \begin{bmatrix}
			v^{ii} & \bs{\lambda}_i^\tp \\
			\bs{\lambda}_i & \bs{\Lambda}_{[i]}
		\end{bmatrix} - \begin{bmatrix}
			v^{ii} & \bs{\lambda}_i^\tp \\
			\bs{\lambda}_i & \bs{\lambda}_i\bs{\lambda}_i^\tp/v^{ii} 
		\end{bmatrix}\\
		&=&  \begin{bmatrix}
			0& \bs{0} \\
			\bs{0} & \bs{\Lambda}_{[i]} - \bs{\lambda}_i\bs{\lambda}_i^\tp/v^{ii} 
		\end{bmatrix}.
	\end{eqnarray*}
	Proof follows from Proposition~\ref{prop:viinv}.
\end{proof}

Suppose that we model the response with the $i$-th observation deleted as
$$\bs{y}_{[i]} = \bs{X}_{[i]}\bs{\tau}_{[i]} + \bs{Z}_{[i]}\bs{u}_{[i]} + \bs{e}_{[i]}$$
where 
$$
\begin{bmatrix}
\bs{u}_{[i]} \\
\bs{e}_{[i]}
\end{bmatrix} \sim N\left(\begin{bmatrix}
\bs{0} \\
\bs{0}
\end{bmatrix}, 
\begin{bmatrix}
\bs{G} & \bs{0} \\
\bs{0} & \bs{R}_{[i]}
\end{bmatrix}
\right)
$$
where $\bs{R}_{[i]}$ is $\bs{R}$  with the $i$-th row and column removed.  Note that $\bs{V}_{[i]} = \bs{Z}_{[i]}\bs{G}\bs{Z}_{[i]}^\tp + \bs{R}_{[i]}$.
\begin{prop}
	$\hat{\bs{\tau}} = \hat{\bs{\tau}}_{[i]}$. \label{thm:4}
\end{prop}
\begin{proof}
	Using Proposition~\ref{prop:mme} and ~\ref{prop:Pdelta}, we have
	\begin{eqnarray*}
		\hat{\bs{\tau}} &=& \left(\begin{bmatrix}
			\bs{x}_i & \bs{X}_{[i]}^\tp 
		\end{bmatrix}  \begin{bmatrix}
			0 & \bs{0}^\tp \\
			\bs{0} & \bs{V}^{-1}_{[i]}
		\end{bmatrix} \begin{bmatrix}
			\bs{x}_i^\tp \\
			\bs{X}_{[i]}
		\end{bmatrix}\right)^-\begin{bmatrix}
			\bs{x}_i & \bs{X}_{[i]}^\tp 
		\end{bmatrix}\begin{bmatrix}
			0 & \bs{0}^\tp \\
			\bs{0} & \bs{V}^{-1}_{[i]}
		\end{bmatrix}\begin{bmatrix}
			y_i \\
			\bs{y}_{[i]}
		\end{bmatrix}\\
		&=& (\bs{X}_{[i]}^\tp\bs{V}_{[i]}^{-1}\bs{X}_{[i]})^{-}\bs{X}_{[i]}^\tp\bs{V}_{[i]}^{-1}\bs{y}_{[i]}.
	\end{eqnarray*} \\
	
\end{proof}


\bibliographystyle{plainnat} 
\bibliography{ref} 

\end{document}